\documentclass[a4paper,10pt]{amsart}
\usepackage{graphicx}
\usepackage{amssymb,amsmath,amstext,amscd}

\theoremstyle{plain}
\newtheorem{thm}{Theorem}[section]
\newtheorem{proposition}{Proposition}[section]
\newtheorem{lemma}{Lemma}[section]
\newtheorem*{corollary}{Corollary}
\newtheorem*{main theorem}{Theorem}

\theoremstyle{definition}

\newtheorem*{definition}{Definition}

\begin{document}

\title{A Note on non-compact Cauchy surface}
\author{Do-Hyung Kim}
\address{National Institute for Mathematical Sciences(NIMS)
(Tower Koreana Building 301, 385-16, Doryong-dong, Yuseong-Gu,
Daejeon, Korea, 305-340)} \email{dhkim@nims.re.kr}


\begin{abstract}
It is shown that if a space-time has non-compact Cauchy surface,
then its topological, differentiable, and causal structure are
completely determined by a class of compact subsets of its Cauchy
surface. Since causal structure determines its topological,
differentiable, and conformal structure of space-time, this gives a
natural way to encode the corresponding structures into its Cauchy
surface.
\end{abstract}

\maketitle

\section{Introduction} \label{section: 1}
 By causality we refer to a general question of which events can
 influence a given event and it is one of the major areas in
 Lorentzian geometry, which gives mathematical tools for global
 analysis in our space-time. Initially, attention was focused on the
 concept of causality violation since, in 1949, G$\ddot{o}$del
 proposed a solution of Einstein's field equation, which contains a closed causal
 curve.(Ref. \cite{Godel})
 In 1964, when Zeeman (Ref. \cite{Zeeman}) has shown that the causal
 structure of Minkowski space-time already implies it linear
 structure, a new direction of investigation was pointed out. In a
 particular space-time, causality may be trivial, but under fairly
 mild conditions, it is closely related to fundamental geometric
 properties. For example, it is well-known that the Alexandrov
 topology $\mathcal{A}$, which is defined only in terms of causality,
 agrees with the given manifold topology if and only if strong
 causality condition holds. Furthermore, it is known that if two
 space-times have the same causal structures, then they are
 diffeomorphic and isometric up to a conformal factor.(Ref.
 \cite{Malament}, \cite{Fullwood}, \cite{HKM})

 In this sense, causal structure of our space-time gives the most
 important information of our space-time. Among the hierarchy of
 causality conditions, the most stringent condition is global
 hyperbolicity, of which the definition is that there exists a
 hypersurface such that every inextendible timelike curve must meet
 the hypersurface exactly once. We call such a surface as a Cauchy
 surface. In fact, global hyperbolicity was invented for dealing
 with hyperbolic differential equations on a manifold and is the
 natural condition to impose to ensure the existence and uniqueness
 of solutions of hyperbolic equations.(Ref. \cite{Leray} and
 Chapter 10 of Ref. \cite{Wald}) This
 condition has also been used in the study of gravitational fields
 and in initial value problems.

 In view of this, since any timelike curves must meet Cauchy surface
 exactly once, it is natural to expect that the information of basic
 structures of a given space-time can be encoded into its Cauchy surface in
 a way. In this paper, it is shown that, if a space-time is globally
 hyperbolic with non-compact Cauchy surface, then its topological,
 differentiable, and causal structure can be encoded into its Cauchy
 surface. By use of this, it is shown that, for given two
 space-times with non-compact Cauchy surfaces $\Sigma$ and
 $\Sigma^\prime$, if there exists a function $f : \Sigma \rightarrow
 \Sigma^\prime$ such that $f$ preserves the information which will be
 defined in this paper, then the two space-times are homeomorphic,
 diffeomorphic, and they are conformally isometric. In conclusion, since we can
 assume that Cauchy surface is spacelike hypersurface, we can say that the information of spacelike
 structure can determine the whole structure of its space-time, if the space-time
 has a non-compact Cauchy surface.

\section{Basics on causality theory} \label{section: 2}

By a space-time $M$, we mean a smooth, connected, Hausdorff
$n$-dimensional Lorentzian manifold with signature $(-,+,\cdots,+)$.
We also assume that $M$ is time-oriented. We define $v \in T_pM$ to
be timelike (null, spacelike, resp.) if its inner product with
itself is less than (equal to, greater than, resp.) zero. A smooth
curve with non-zero tangent is said to be timelike if its tangent is
everywhere timelike and causal if the tangent is timelike or null.
If there is a future-directed timelike curve from $p$ to $q$, we
write $p \ll q$ and we say that $q$ is a chronological future of $p$
or $p$ is a chronological past of $q$. If there is a future-directed
causal curve from $p$ to $q$, then we write $p \leq q$ and we say
that $q$ is a causal future of $p$ or $p$ is a causal past of $q$.
The relation ``$p \leq q$ but not $p \ll q$'' is written $p
\rightarrow q$ and is termed horismos. We define a causal curve to
be future-inextendible if it has no future endpoint and to be
past-inextendible if it has no past endpoint. We define a subset $S
\subset M$ to be achronal if no two points are chronologically
related. The following is the basic property of causal relations.

\begin{proposition}
(i) If $x \leq y$ and $y \ll z$, then $x \ll z$.\\
(ii) If $x \ll y$ and $y \leq z$, then $x \ll z$.\\
(iii) If $x \rightarrow y$, then there is a null geodesic from $x$
to $y$ without conjugate points.\\
\end{proposition}

\begin{definition}
The set $I^+(x) = \{ y \in M ~|~ x \ll y\}$ is called the
chronological future of $x$; $I^-(x) = \{ y \in M ~|~ y \ll x\}$ is
the chronological past of $x$; $J^+(x) = \{ y \in M ~|~ x \leq y\}$
is the causal future of $x$; $J^-(x) = \{ y \in M ~|~ y \leq x \}$
is the causal past of $x$. The chronological or causal future of a
set $S \subset M$ is defined by $\displaystyle I^+(S)= \bigcup_{x
\in M} I^+(x)$, $\displaystyle J^+(S) = \bigcup_{x \in M} J^+(x)$,
respectively.
\end{definition}

It is known that $I^+(p)$ is always open and $\partial{J^-(p)}$ is
achronal of which the proof can be found in Ref. \cite{Penrose1972},
\cite{Beem}, \cite{Oneill}. We turn next to the possibilities
involving causality violations in space-times. If $M$ contains a
closed timelike curve, then an observer could take a trip from which
he returns before his departure. In fact, it is known that causality
condition has a intimate relation with the topology of space-time.
For example, we can show that any compact space-time must have a
closed timelike curve by the use of the fact that $I^+(S)$ is always
open. In connection with this, a number of causality conditions have
been introduced in general relativity and we present some of them
which will be used in this paper.
 Space-times which do not contain any closed timelike curve are said to be chronological.
 A space-time with no closed nonspacelike curves is said to be
 causal. A space-time is said to be distinguishing if for all points
 $p$ and $q$ in $M$, either $I^+(p)=I^+(q)$ or $I^-(p)=I^-(q)$
 implies $p=q$. In a distinguishing space-time, distinct points have
 distinct chronological futures and chronological pasts. Thus,
 points are distinguished both by their chronological futures and
 pasts. A space-time is said to be strongly causal at $p$ if $p$ has
 arbitrarily small neighborhood $U$ such that no nospacelike curve
 intersects $U$ in a disconnected set. Note that if $M$ is
 distinguishing and $p \in M$, then $p$ has an arbitrarily small
 neighborhood $U$ of $p$ such that any timelike curve from $p$ can
 not intersects $U$ in a disconnected set, whereas if $M$ is strongly
 causal, then any timelike curve starting from any points in $U$
  can not intersect $U$ in a disconnected set.

  Since $I^+(p)$ is open, $I^+(p) \cap I^-(q)$ is open for any $p$
  and $q$ in $M$. We can show that the sets of the form $I^+(p) \cap
  I^-(q)$ defines a basis for a topology on $M$, which we will call the
  Alexandrov topology of $M$.If a space-time $M$ is strongly causal, then we can state the close
connection with the topology of $M$ and its causal relation as the
following well-known theorem shows.

\begin{thm}
The followings are equivalent.\\
(i) $M$ is strongly causal.\\
(ii) The Alexandrov topology agrees with the manifold topology.\\
(iii) The Alexandrov topology is Hausdorff.\\
\end{thm}
\begin{proof}
See Theorem 4. 24. in \cite{Penrose1972}.
\end{proof}

We say that $I^+$ is inner continuous at $p \in M$ if for each
compact set $K \subset I^+(p)$, there exists a neighborhood $U$ of
$p$ such that $ K \subset I^+(q)$ for each $q \in U$. We also define
the set-valued function $I^+$ to be outer continuous at $p \in M$ if
for each compact set $K \subset M - \overline{I^+(p)}$, there exists
some neighborhood $U$ of $p$ such that $K \subset M -
\overline{I^+(q)}$ for each $q \in U$.

A distinguishing space-time is said to be causally continuous if the
both $I^+$ and $I^-$ are outer continuous. Since $I^+$ and $I^-$ are
always inner continuous (Ref. \cite{HS}), the causally continuous
space-times are those distinguishing space-times for which both the
chronological future and past of a point vary continuously with the
point.

One of the most important causality condition which we will discuss
in this paper, is global hyperbolicity. Globally hyperbolic
space-times have the important property, frequently invoked during
specific geodesic constructions, that any pair of causally related
points can be joined by a causal geodesic with maximal length.

\begin{definition}
A distinguishing space-time $M$ is causally simple if $J^+(p)$ and
$J^-(p)$ are closed subsets of $M$ for all $p \in M$. A space-time
$M$ is globally hyperbolic if $M$ is strongly causal and $J^+(p)
\cap J^-(q)$ is compact for any $p$ and $q$ in $M$.
\end{definition}

We summarize the relations between the various causality relations
in the following theorem. \cite{Penrose1972}, \cite{Beem},
\cite{Oneill}

\begin{thm}
Let $M$ be a space-time. Then the followings hold.\\
(i) If $M$ is globally hyperbolic, then $M$ is causally simple.\\
(ii) If $M$ is causally simple, then $M$ is causally continuous.\\
(iii) If $M$ is causally continuous, then $M$ is strongly causal.\\
(iv) If $M$ is strongly causal, then $M$ is distinguishing.\\
(v) If $M$ is distinguishing, then $M$ is chronological.\\
\end{thm}

It is known that a space-time $M$ is globally hyperbolic if and only
if there exists an $(n-1)$ dimensional hypersurface $\Sigma$ such
that every inextendible timelike curve meets $\Sigma$ exactly once.
This implies that the surface must be achronal and we call the
surface a Cauchy surface. Though a Cauchy surface must be achronal,
it is shown that if $M$ is globally hyperbolic, then $M$ has a
spacelike Cauchy surface. Since we assume that $M$ is globally
hyperbolic with non-compact Cauchy surface $\Sigma$ in this paper,
we can assume that the non-compact Cauchy surface $\Sigma$ is
spacelike.

\begin{definition}
Let $S$ be an achronal subset of $M$. The edge of $S$ is a set of
points $x$ such that for every neighborhood $U$ of $x$, there are
two points $y$ and $z$ in $U$ and two timelike curves in $U$ from
$y$ to $z$, just one of which meets $S$.
\end{definition}

By the above definition, it is easy to see that if $S$ is achronal,
we have $\overline{S}-S \subset edge S \subset \overline{S}$.

\begin{proposition} \label{nbd-edge}
Let $S \subset M$ be achronal. Then $p \notin edge S$ if and only if
there is a neighborhood $U$ of $p$ such that $S \cap U$ is an
achronal boundary in the space-time manifold $U$.
\end{proposition}

\begin{proof}
See Proposition 5.8 in \cite{Penrose1972}
\end{proof}

By the above proposition, it is easy to prove the following.

\begin{corollary}
If $S$ is achronal, then $edge S$ is closed.
\end{corollary}

\begin{proposition} \label{edge-closed}
Let $S$ be achronal, then the followings hold.\\
 (i) $S$ is a
topological hypersurface if and only if $S \cap edge S =
\emptyset$.\\
 (ii) $S$ is a closed topological hypersurface if and only if $edge
 S$ is empty.\\
\end{proposition}
\begin{proof}
See Chapter 14 in \cite{Oneill}
\end{proof}

\begin{proposition} \label{proposition-edge}
Let $x$ be in $\partial{I^-(S)}-\overline{S}$. Then there exists a
null geodesic on $\partial{I^-(S)}$ with past endpoint $x$, which is
either future-inextendible or has a future endpoint on
$\overline{S}$.
\end{proposition}
\begin{proof}
Since $x \in \partial{I^-(S)}$, we can choose a sequence $\{x_i\}$
in $I^-(S)$ which converges to $x$. Since $\{x_i\} \in I^-(S)$, we
can choose a sequence $\{p_i\}$ in $S$ such that there is a timelike
curve $\gamma_i$ from $x_i$ to $p_i$. Let $\gamma$ be a limit curve
of $\gamma_i$. Then $\gamma$ is a future directed causal curve from
$x$ to $\overline{S}$. If $\gamma$ is timelike at a point, then $x
\in I^-(\overline{S}) = I^-(S)$, which  contradicts to $x \in
\partial{I^-(S)}$. Thus $\gamma$ is a null geodesic.
Suppose that $\gamma$ is not inextendible. Then, it has a future
endpoint $y$ in $\partial{I^-(S)}$ since $\partial{I^-(S)}$ is
closed. If $y \notin  \overline{S}$ we can apply the above argument
to get another null geodesic which does not continue $\gamma$. Since
these two null geodesics have different directions, this contradicts
to the achronality of $\partial{I^-(S)}$.
\end{proof}

\section{Causally admissible slice} \label{section: 3}

Throughout this section we assume that $M$ is a space-time with a
non-compact Cauchy surface $\Sigma$.

\begin{definition}
A subset $S$ of a Cauchy surface $\Sigma$ is called a future
admissible set or future admissible slice (past admissible,
respectively) if there is $p$ in $M$ such that $S = J^-(p) \cap
\Sigma$ ($S=J^+(p) \cap \Sigma$, respectively). We call future
admissible sets and past admissible sets as causally admissible
sets. If $S_p = J^-(p) \cap \Sigma$ ($S_p = J^+(p) \cap \Sigma$,
respectively), we call $p$ a representative point of $S_p$.
\end{definition}

Since the same properties hold for past admissible sets, we only
investigate the properties of future admissible sets in this
section. For given $p \in J^+(\Sigma)$, let $S_p = J^-(p) \cap
\Sigma$ and $S_p^{\circ} = I^-(p) \cap \Sigma$. Then, since $M$ is
globally hyperbolic with non-compact Cauchy surface, $S_p$ is a
compact subset of $\Sigma$ and $\Sigma - S_p \neq \emptyset$ for any
$p \in M$. Since $edge S$ is a closed subset of compact $S_p$, $edge
S$ is also a compact subset of $M$. We state some properties of
$S_p$ and $S^{\circ}_p$ which will be used later.

\begin{proposition} \label{coro}
If $M$ is globally hyperbolic with a Cauchy surface $\Sigma$, and $S
\subset \Sigma$, then $edge S$ is the same with the set of boundary
points of $S$ in $\Sigma$.
\end{proposition}
\begin{proof}
The proof of $edge S \subset \partial{S}$ is almost trivial. Let $x$
be a boundary point of $S$ in $\Sigma$. Let $U$ be a neighborhood of
$x$ in $M$. Since $M$ is strongly causal, we can choose an
Alexandrov basis $I^+(p) \cap I^-(q)$ of $x$ such that $I^+(p) \cap
I^-(q) \subset U$. Since $x$ is a boundary point of $S$ in $\Sigma$,
we can choose $y_1 \in S$ and $y_2 \in \Sigma - S$ in $I^+(p) \cap
I^-(q)$. The existence of a timelike curve from $p$ to $q$ through
$y_1$ and a timelike curve from $p$ to $q$ through $y_2$ implies
that $x \in edge S$.
\end{proof}

From the above proposition, we can see immediately that $edge S =
edge (\Sigma-S)$ and $edge S_p = \emptyset$ if and only if $S_p =
\Sigma$. We also state some corollaries.

\begin{corollary}
If $\Sigma$ is non-compact, then $edge S_p$ and $edge S_p^\circ$ are
not empty for each $p \in I^+(\Sigma)$.
\end{corollary}
\begin{proof}
If $edge S_p = \emptyset$, then by the above proposition, we have
$\partial{S_p} = \emptyset$ and thus we have $S_p = \Sigma$, which
is a contradiction since $S_p$ is compact and $\Sigma$ is not. If
$edge S_p^\circ = \emptyset$ for some $p \in I^+(\Sigma)$, then
$S_p^\circ$ is closed by (ii) of Proposition ~\ref{edge-closed}.
Since $S_p^\circ$ is also open and $\Sigma$ is connected, we have
$S_p^\circ = \Sigma$. This implies that $S_p = \Sigma$, which is a
contradiction.
\end{proof}

\begin{corollary}
$edge S_p \subset \partial{I^-(p)} = \partial{J^-(p)}$
\end{corollary}
\begin{proof}
Let $x$ be in $edge S_p$ and $U$ be any neighborhood of $x$ in $M$.
Since $edge S_p \subset \overline{S_p} = S_p$, $U$ contains a point
of $J^-(p)$. Since $M$ is globally hyperbolic, the closure of
$I^-(p)$ is the same with $J^-(p)$. Thus, $U$ contains a point in
$I^-(p)$. By the previous proposition, $x$ is a boundary point of
$S_p$ in $\Sigma$ and so $U$ contains a point of $\Sigma - S_p$,
which is not in $I^-(p)$.
\end{proof}

\begin{lemma} \label{edge-geodesic}
If $x \in edge S_p$, then $x \rightarrow p$ and thus there exists a
future-directed null geodesic from $x$ to $p$.
\end{lemma}
\begin{proof}
 Since $S_p$ is closed, we have $x \in S_p$ and $x
\leq p$. If $x \ll p$, then there exists a neighborhood $U$ of $x$
such that $y \ll p$ for all $y \in U$ since the relation $\ll$ is
open. This implies that $U \cap \Sigma$ is an open neighborhood of
$x$ in $\Sigma$ such that $U \cap \Sigma \subset S_p$, which is a
contradiction to the fact that $x \in edge S$. Thus, $x \rightarrow
p$ and we have a desired null geodesic.
\end{proof}

In Proposition ~\ref{proposition-edge}, we have seen that the
boundary of a past set is generated by null geodesics. When $M$ is
globally hyperbolic, we can state the result more precisely as the
following.

\begin{lemma} \label{endgeodesic}
Let $M$ be globally hyperbolic with a non-compact Cauchy surface
$\Sigma$ and $x \in
\partial{I^-(S)} - \overline{S}$ with $S \subset \Sigma$. Then there exists a future-directed null geodesic
from $x$ with future endpoint in $edge S$.
\end{lemma}
\begin{proof}
By Corollary 3. 32 in Reference \cite{Beem} and by Proposition
\ref{proposition-edge}, there is a future-directed null geodesic
$\gamma$ from $x$ to $y \in \overline{S}$. It remains to show that
$y \in edge S$. Assume that $y \notin edge S$. Then by the
definition of edge, we can choose a neighborhood $U$ of $y$ such
that if $\alpha$ is a timelike curve in $U$ from $z_1 \in U$ to $z_2
\in U $, then any other timelike curve from $z_1$ to $z_2$ must meet
$S$ if and only if $\alpha$ does. If we choose a point $z_1$ in $U
\cap I^-(y)$, then since $y \in \overline{S}$, an open neighborhood
$I^+(z_1) \cap I^-(z_2)$ of $y$ contains a point in $S$ for some
suitably chosen $z_2 \in U \cap I^+(y)$. Then we get the timelike
curve $\alpha$ from $z_1$ to $z_2$ through a point in $S$.
Therefore, every timelike curve in $U$ from $z_1$ to $z_2$ must meet
$S$. Since $y$ is the future endpoint of $\gamma$, there is $t_0$
such that $\gamma(t_0) \in I^+(z_1) \cap I^-(z_2)$ and $\gamma(t_0)
\neq y$. If we consider a future directed timelike curve from $z_1$
to $\gamma(t_0)$ followed by a timelike curve from $\gamma(t_0)$ to
$z_2$, then it must meet $S$ and we have $x \in I^-(S)$. This is a
contradiction to the fact that $x \in
\partial{I^-(S)}$.
\end{proof}

If $S$ is a future admissible set, then there is a point $p$ such
that $S = J^-(p) \cap \Sigma$. It is natural to ask whether such a
realizing point $p$ can be uniquely determined. It is easy to see
that in two-dimensional Einstein's static universe which has a
compact Cauchy surface, there are infinitely many points such that
$J^-(p) \cap \Sigma$ is the whole of $\Sigma$. However, the next
proposition tells us that such a point can be uniquely determined if
$\Sigma$ is non-compact.

\begin{proposition} \label{sp-unique}
If $\Sigma$ is non-compact and $S_p = S_q$, then $p = q$.
\end{proposition}
\begin{proof}
We assume that $p \neq q$ and denote $S_p = S_q$ by $S$. Since
$I^-(S) \neq I^-(\Sigma)$, we can choose $y \in \partial{I^-(S) -
\overline{S}}$ and we must have $y \in \partial{I^-(p)}$ and $y \in
\partial{I^-(q)}$ since $S_p = S_q$. Then by Lemma ~\ref{endgeodesic},
there is a null geodesic $\gamma_x$ from $y$ to $x \in edge S$ which
generates $\partial{I^-(S)}$. Since $x \in edge S$, by Lemma
~\ref{edge-geodesic}, there are two null geodesics $\gamma_p$ and
$\gamma_q$ from $x$ to $p$ and $q$, which generates
$\partial{I^-(p)}$ and $\partial{I^-(q)}$, respectively. Let us
assume that the tangent vectors of $\gamma_p$ and $\gamma_q$ have
different directions at $x$. If $\gamma_x$ and $\gamma_p$ does not
constitute a single null geodesic, then we have $y \ll p$ which
contradicts to the fact that $y \in
\partial{I^-(S)}$. Thus $\gamma_x$ and $\gamma_p$ constitute a
single null geodesic. However, this implies that $\gamma_x$ and
$\gamma_q$ does not constitute a single null geodesic since
$\gamma_p$ and $\gamma_q$ have different directions at $x$ and thus
we have $y \ll q$. This is a contradiction to $y \in
\partial{I^-(q)}$. Therefore, $\gamma_p$ and
$\gamma_q$ must have the same direction at $x$. Since the two curves
are geodesics with the same initial direction, we can treat them as
the same null geodesic and we denote the geodesic by $\gamma$. By
the same argument we can choose another $x^\prime \in edge S$ and
two null geodesics $\gamma_p^\prime$ and $\gamma_q^\prime$ from
$x^\prime$ to $p$ and $q$. By the same reason as above,
$\gamma_p^\prime$ and $\gamma_q^\prime$ must have the same initial
direction. We denote the null geodesic by $\gamma^\prime$. Without
loss of generality, we can assume that $\gamma$ and $\gamma^\prime$
meet at $p$ first. If we follow the curve $\gamma$ from $x$ to $p$
and the curve $\gamma^\prime$ from $p$ to $q$, we must have that $x
\ll q$ since $\gamma$ and $\gamma^\prime$ have different directions
at $p$. This again gives a contradiction to the fact that $x \in
\partial{I^-(q)}$. This contradiction stems from the assumption that
$p \neq q$ and we must conclude that $p = q$
\end{proof}

\begin{proposition}
$S_p = \overline{S_p^{\circ}}$
\end{proposition}
\begin{proof}
Clearly, $S_p^{\circ} \subset S_p$ and thus $\overline{S_p^{\circ}}
\subset \overline{S_p} = S_p$ since $S_p$ is closed. Conversely we
show that $S_p \subset \overline{S_p^{\circ}}$. Let $x \in S_p$,
then $x \leq p$. If $x \ll p$, then $x \in S_p^{\circ} \subset
\overline{S_p^{\circ}}$. If $x$ and $p$ does not satisfy $x \ll p$,
then we have $x \rightarrow p$. This means that $x \in
\partial{J^-(p)}$. Since $M$ is globally hyperbolic, we have $J^-(p)
= \overline{I^-(p)}$ and thus any neighborhood of $x$ contains a
point of $I^-(p)$. Let $U$ be any neighborhood of $x$ such that no
causal curves intersect $U$ in a disconnected set and $I^+(y) \cap
I^-(z)$ be an Alexandrov neighborhood of $x$ such that $I^+(y) \cap
I^-(z) \subset U$. Then $I^+(y) \cap I^-(z)$ contains a point
$y^\prime$ in $I^-(p)$ and we get a timelike curve $\gamma$ from $p$
to $y$ through $y^\prime$. Since $y \in I^-(x)$ and $x \in \Sigma$,
$\gamma$ must meet $\Sigma$ in $I^+(y) \cap I^-(z) \subset U$. This
means that $U \cap S_p^{\circ} \neq \emptyset$.
\end{proof}

In Proposition ~\ref{sp-unique}, we have shown that the
representative point of a future causally admissible set is unique.
By the above proposition, we can see immediately that the same
result holds for $S^{\circ}_p$.

\begin{corollary}
If $S^{\circ}_p = S^{\circ}_q$, then $p = q$.
\end{corollary}
\begin{proof}
By taking the closure in both sides of  $S^{\circ}_p = S^{\circ}_q$,
we get $S_p = S_q$ by the above proposition. By Proposition
~\ref{sp-unique}, we have $p = q$.
\end{proof}

 In Proposition ~\ref{coro}, we have shown that $edge S_p$ is the
same with the set of boundary points of $S_p$ in $\Sigma$. Since
$S_p = \overline{S_p^{\circ}}$ the above proposition implies that
$S_p^{\circ} \cup edge S_p = S_p$ and $S_p - edge S_p = S_p^{\circ}$
in $\Sigma$. Thus, we have the following proposition.

\begin{proposition}
If $x \in S_p - edge S_p$, then there exists a future directed
timelike curve from $x$ to $p$.
\end{proposition}
\begin{proof}
By the above argument, we have $S_p - edge S_p = S_p^{\circ}$ and
the result follows.
\end{proof}

\begin{lemma} \label{timesp}
Let $A$ and $B$ be proper subsets of a topological space. If $A
\subset B$ and $\partial{A} \cap \partial{B} = \emptyset$, then $A
\subset int B$.
\end{lemma}
\begin{proof}
Assume that $\exists x \in A - intB$. Then for any neighborhood $U$
of $x$ contains a point $z \notin B$ by the definition of interior.
Since $x \in A \subset B$, we have $x \in \partial{B}$. Since $z
\notin B$ and $A \subset B$, we have $z \notin A$ and thus $x \in A$
implies that $x \in
\partial{A}$. This contradicts to $\partial{A} \cap \partial{B} =
\emptyset$.
\end{proof}

\begin{proposition} \label{time-sp}
If $S_p \subset S_q$ and $edge S_p \cap edge S_q = \emptyset$, then
$S_p \subset S_q^{\circ}$. In other words, every point in $S_p$ can
be connected by a timelike curve to $q$.
\end{proposition}
\begin{proof}
Since $edge S_p \cap edge S_q = \emptyset$, Proposition ~\ref{coro}
gives $\partial{S_p} \cap \partial{S_q} = \emptyset$ in the space
$\Sigma$. By Lemma ~\ref{timesp}, we have $S_p \subset intS_q =
S_q^{\circ}$.
\end{proof}

\section{Properties of causally admissible slices} \label{section: 4}

In the previous section, we have defined and investigated some
properties of causally admissible sets, which are the building
blocks for encoding the geometric information of space-time into its
Cauchy surface. In this section, we show how the causally admissible
sets can be used to encode the information. By the same reason, we
investigate and state the properties of future admissible slices
since the same properties hold for past admissible slices.
\begin{proposition} \label{thm1}
If $p \leq q$, then $S_p \subset S_q$.
\end{proposition}
\begin{proof}
Let $x \in S_p$ then, by definition, we have $x \leq p$. Since $p
\leq q$, by the transitivity of the relation $\leq$, we have $x \leq
q$. This implies that $x \in S_q$ and the proof is completed.
\end{proof}

For chronological relation $\ll$, we have one more condition that
$edge S_p \cap edge S_q = \emptyset$, as the following proposition
shows.

\begin{proposition} \label{thm3}
 If $p \ll q$, then $S_p \subset S_q$ and $edge S_p \cap edge S_q
 = \emptyset$
\end{proposition}
\begin{proof}
 It is clear that $S_p \subset S_q$. If $S_p = S_q$, then $S_p =
 I^-(q) \cap \Sigma$ since $S_p = J^-(p) \cap \Sigma \subset I^-(q)
 \cap \Sigma \subset J^-(q) \cap \Sigma = S_q$. This implies that
 $S_p$ is open in $\Sigma$. Since $S_p$ is closed in $\Sigma$ and $\Sigma$ is
 connected, we must have $S_p = \Sigma$ which is a contradiction.

If there exists $x$ in $edge S_p \cap edge S_q$, then by Proposition
~\ref{edge-geodesic}, we have $x \rightarrow p$ and $x \rightarrow
q$. However, since $p \ll q$, we have $x \ll q$ which contradicts to
$x \rightarrow q$.
\end{proof}

By the above Proposition and Proposition \ref{time-sp}, we can see
that the set $S_p$ strictly increases as $p$ goes to its
chronological future.

One of the goals of this section is to show that the converses of
the above two propositions are also satisfied. To this end, we need
another tool, the concept of domain of dependence.

\begin{definition}
The future domain of dependence of $S$, $D^+(S)$, is the set of all
points $p$ such that every past-inextendible timelike curve from $p$
intersects $S$. The past domain of dependence $D^-(S)$, is defined
in a similar way by interchanging the roles of past and future. The
domain of dependence of $S$ is the set $D(S)=D^+(S) \cup D^-(S)$.
The future Cauchy horizon of $S$, $H^+(S)$, is defined as
$H^+(S)=\{x \,\, | \,\, x \in D^+(S) \,\, \mbox{and} \,\, I^+(x)
\cap D^+(S) = \emptyset \}$. The past Cauchy horizon, $H^-(S)$, is
defined in a similar way by interchanging the roles of past and
future. The Cauchy horizon of $S$ is the set $H(S)=H^+(S) \cup
H^-(S)$.
\end{definition}

In terms of domain of dependence, it is easy to see that $M$ is
globally hyperbolic if and only if there exists an achronal
hypersurface $\Sigma$ such that $D(\Sigma)=M$. The followings are
well-known properties of domain of dependence and the proof can be
found in Ref.\cite{Penrose1972}, \cite{Oneill}.

\begin{proposition} \label{ds1}
Let $S$ be achronal and closed. Then the followings hold.\\
(i) $D^+(S)$ is closed.\\
(ii) $H^+(S)$ is achronal and closed.\\
(iii) $\partial D^+(S) = H^+(S) \cup S$.\\
(iv) $\partial D(S) = H(S)$.\\
(v) $I^+(edgeS) \cap D^+(S) = \emptyset$.\\
\end{proposition}

We now investigate relations between $J^-(p)$ and $D^+(S_p)$.

\begin{proposition} \label{hsp}
If we let $S_p=J^-(p) \cap \Sigma$, then $p \in H^+(S_p)$.
\end{proposition}
\begin{proof}
It is easy to see that $p \in D^+(S)$ and we omit the proof. Assume
that $p \notin H^+(S_p)$. Then there exists $x$ in $D^+(S_p)$ such
that $p \ll x$. Since $p \ll x$, by Proposition \ref{thm3} and
Proposition \ref{time-sp}, we have $S_p \subset S_x^{\circ} \subset
S_x$. In other words, there exists a past-directed timelike curve
from $x$ that does not meet $S_p$, which is a contradiction to $x
\in D^+(S_p)$.
\end{proof}

\begin{proposition} \label{ds2}
If $M$ is globally hyperbolic with Cauchy surface $\Sigma$, then
$D^+(S)\cap I^+(S)=I^+(S)-I^+(edgeS)$ for any $S \subset \Sigma$.
\end{proposition}
\begin{proof}
To show the inclusion $D^+(S)\cap I^+(S) \subset I^+(S)-I^+(edgeS)$
is easy from (v) in Proposition \ref{ds1}. To show the reverse
inclusion, let us assume that there exists $x$ in $I^+(S)$ such that
$x \notin D^+(S)$. Then, there exists a past-inextendible timelike
curve $\gamma$ from $x$ that does not meet $S$. Since $\gamma$ is
inextendible and $M$ is globally hyperbolic, $\gamma$ must meet
$\Sigma$ at a point in $\Sigma - S$. Thus, $\gamma$ must meet
$\partial I^+(S)$ at $y$, say. By Lemma \ref{endgeodesic}, there
exists a null geodesic from $y$ to a point in $edge S$. Then the
curve $\gamma$ from $x$ to $y$ joined with null geodesic from $y$ to
the point in $edge S$ implies that $x \in I^+(edge S)$ and the proof
is completed.
\end{proof}

\begin{lemma} \label{closed}
If $M$ is globally hyperbolic, then the relation $\leq$ is closed.
\end{lemma}
\begin{proof}
See the Lemma 22 of Chapter 14 in Ref \cite{Oneill}.
\end{proof}

\begin{proposition} \label{hs}
Let $x \in H^+(S)-S$ with $S=S_p$ for some $p \in I^+(\Sigma)$ where
$\Sigma$ is a Cauchy surface. Then there exists a null geodesic on
$H^+(S)$ from $x$ to a point in $edge S$.
\end{proposition}
\begin{proof}
Let $\{x_i\}$ be a sequence in $I^+(x)$ such that $x_{i+1} \ll x_i$
and $x_i$ converges to $x$. Then, since $x \in H^+(S)$, $x_i \notin
D^+(S)$. By Proposition \ref{ds2}, for each $x_i$, there exists
$y_i$ in $edge S$ such that $y_i \ll x_i$. Since $edge S$ is
compact, without loss of generality, we can assume that there exists
a point $y \in edge S$, such that $y_i$ converges to $y$. Since $y_i
\ll x_i$ and $M$ is globally hyperbolic, we have $y \leq x$ by the
above lemma. If $y \ll x$, then $x \in I^+(edge S)$, which
contradicts to $x \in D^+(S)$. Thus, we must have $y \rightarrow x$
and we get a null geodesic $\eta$. It remains to show that every
point on $\eta$ lies on $H^+(S)$. Let $z$ be a point on the null
geodesic $\eta$. Let $U$ be a neighborhood of $z$ that lies in
$I^+(S)$. Then any chronological future point of $z$ in $U$ must be
in $I^+(edge S)$, which is is not in $D^+(S)$ by Proposition
\ref{ds2}. Let $z^{\prime}$ be a point in $U$ such that $z^{\prime}
\ll z$. Then $z \rightarrow x$ implies that $z^{\prime} \ll x$. Let
$\gamma$ be a timelike curve from $x$ to $z^{\prime}$. Then, since
$z^{\prime} \in U \subset I^+(S)$, $\gamma$ does not meet $\Sigma$.
Thus any past inextendible timelike curve $\gamma^{\prime}$ from
$z^{\prime}$ must meet $S$ since $\gamma \cup \gamma^{\prime}$ is a
past inextendible timelike curve from $x$ and $x \in H^+(S)$. That
is to say, any neighborhood $U$ of $z$ contains points in $D^+(S)$
and not in $D^+(S)$. Therefore, we have $z \in \partial
D^+(S)=H^+(S) \cup S$, and $z \in H^+(S)$.
\end{proof}

In Proposition \ref{proposition-edge}, we have seen that for $x \in
\partial{J^-(p)}$, there exists a future-directed null geodesic from $x$ to $p$
which generates $\partial{J^-(p)}$. If we extend the generating null
geodesic to the past, we can not assure that the geodesic be on
$\partial{J^-(p)}$ as in the case with two-dimensional Einstein's
static universe. However, in the following theorem, we show that if
the generating null geodesic is extended to the past beyond $x$, it
must reach the point in $edge S_p$ if $\Sigma$ is non-compact, which
plays the key role for the proof of converses of Proposition
\ref{thm1} and Proposition \ref{thm3}.

\begin{thm} \label{thm2}
Let $M$ be globally hyperbolic with non-compact Cauchy surface
$\Sigma$. For $p \in I^+(\Sigma)$, let $S_p = J^-(p) \cap \Sigma$.
Then, for any $x \in \partial{J^-(p)}$, there exists $z^{\prime}$ in
$edge S_p$ such that $z^{\prime} \rightarrow x \rightarrow p$
constitute a single null geodesic.
\end{thm}
\begin{proof}
Choose a sequence $\{p_i\}$ in $I^+(p)$ such that $p_{i+1} \ll p_i$
and $p_i$ converges to $p$. Then by Proposition \ref{hsp}, each $p_i
\notin D^+(S_p)$. Since $x \rightarrow p$ and $p \ll p_i$, we have
$x \ll p_i$ and thus we can choose a past directed timelike curve
$\alpha_i$ from each $p_i$ to $x$. Since $p_i \notin D^+(S_p)$ and
$x \in D^+(S_p)$, the curve $\alpha_i$ must meet $H^+(S_p)$ at, say,
$z_i$. By Proposition \ref{hs}, we get another sequence
$\{z_i^{\prime}\}$ in $edge S_p$ such that $z_i^{\prime} \rightarrow
z_i$. Since $edge S_p$ is compact, there exists a point $z^{\prime}$
in $edge S_p$ such that $z_i^{\prime}$ converges to $z^{\prime}$.
For the sequence $\{z_i\}$, since each $z_i$ is in $J^+(x) \cap
J^-(p_1)$ which is compact, there exists a point $z$ in $J^+(x) \cap
J^-(p_1)$ such that $z_i$ converges to $z$. Since $H^+(S_p)$ is
closed and $z_i \in H^+(S_p)$, we have $z \in H^+(S_p) \subset
D^+(S_p)$. Since $z_i^{\prime} \leq z_i$ and $M$ is globally
hyperbolic, we have $z^{\prime} \leq z$ by Lemma \ref{closed}. If
$z^{\prime} \ll z$, then $z \in I^+(edge S_p)$ which contradicts to
Proposition \ref{ds2} since $z \in D^+(S_p)$. Therefore, we must
have $z^{\prime} \rightarrow z$. Likewise, $z_i \ll p_i$ implies
that $z \leq p$. If $z \ll p$, then this contradicts to the
achronality of $H^+(S_p)$. Therefore we must have $z \rightarrow p$.

We have shown that $z^{\prime} \rightarrow z$ and $z \rightarrow p$.
If the corresponding two null geodesics from $z^{\prime}$ to $z$ and
from $z$ to $p$ have different direction at $z$, then we have
$z^{\prime} \ll p$, which contradicts to Lemma \ref{edge-geodesic}.
Therefore the two null geodesics constitute a single null geodesic.
This implies that $z \in
\partial{J^-(p)}$. If $x = z$, then we have the desired null geodesic
 $z^{\prime} \rightarrow z=x \rightarrow p$.
 Now let us assume that $x \neq z$. Since $x \leq
z_i$ and $M$ is globally hyperbolic, we have $x \leq z$ by Lemma
\ref{closed}. If $x \ll z$, then this contradicts to the achronality
of $\partial{J^-(p)}$ and thus we have $x \rightarrow z$. If two
null geodesics from $x \rightarrow z$ and $z \rightarrow p$ have
different direction at $z$ the, we have $x \ll p$ which contradicts
to $x \rightarrow p$. Thus, two null geodesics $x \rightarrow z$ and
$z \rightarrow p$ constitute a single null geodesic. By the
uniqueness of geodesics, we have $z^{\prime} \rightarrow x
\rightarrow z \rightarrow p$, which constitute a single null
geodesic and the proof is completed.
\end{proof}

Now we are in a position to prove the converse of Proposition
\ref{thm1} and Proposition \ref{thm3}.

\begin{proposition} \label{converse}
If $S_p \subset S_q$ for $p$ and $q$ in $J^+(\Sigma)$, then $p \leq
q$.
\end{proposition}
\begin{proof}
Assume that $p \notin J^-(q)$ and let $\gamma$ be a past-directed
timelike curve from $p$. Then $\gamma$ must meet $S_p^{\circ}$.
Since $S_p^{\circ} \subset S_q^{\circ} \subset S_q$ by hypothesis,
$\gamma$ must meet $\partial{J^-(q)}$ at a point $x$. By Theorem
\ref{thm2}, there exists a past-directed null geodesic from $x$ to a
point in $edge S_q$. By combining the timelike curve $\gamma$ and
the null geodesic, we have $p \in I^+(edge S_q)$. By the definition
of edge, this implies that a past-directed timelike curve from $p$
can reach a point in $\Sigma$ outside $S_q$. However, this
contradicts to $S_p \subset S_q$. Thus we have, $p \in J^-(q)$ and
$p \leq q$.
\end{proof}

By combining Proposition \ref{thm1} and Proposition \ref{converse},
we have the following theorem.

\begin{thm} \label{causal}
Let $p$ and $q$ be in $J^+(\Sigma)$.\\ Then $p \leq q$ if and only
if $S_p \subset S_q$.
\end{thm}

Now we can prove how the causally admissible slices of $\Sigma$
determine timelike relations in $J^+(\Sigma)$.

\begin{thm} \label{timelike}
Let $p$ and $q$ be in $J^+(\Sigma)$.\\ Then $p \ll q$ if and only if
$S_p \subset S_q$ and $edge S_p \cap edge S_q = \emptyset$.
\end{thm}
\begin{proof}
The ``only if" part is proved in Proposition \ref{thm3}. It remains
to show the ``if" part. By Proposition \ref{converse}, $S_p \subset
S_q$ implies that $p \leq q$. If $p \rightarrow q$, then we have $p
\in
\partial{J^-(q)}$ and by Theorem \ref{thm2} again, there exists $z
\in edge S_q$ such that two null geodesics $z \rightarrow p$ and $p
\rightarrow q$ constitute a single null geodesic. However, this
implies that $z \in edge S_p \cap edgeS_q$, which is a
contradiction. This contradiction stems from the assumption that $p
\rightarrow q$. Therefore, we must conclude that $p \ll q $.
\end{proof}

We next investigate how the causally admissible slices determine the
horismos relations in $J^+(\Sigma)$.

\begin{proposition} \label{horis1}
Let $p$ and $q$ be in $J^+(\Sigma)$ and such that $p \neq q$ and $p
\rightarrow q$, then $S_p \subset S_q$ and $edge S_p \cap edge S_q$
has only one element.
\end{proposition}
\begin{proof}
It is obvious that $S_p \subset S_q$. If $edge S_p \cap edge S_q =
\emptyset$, then by Theorem \ref{timelike}, we have $p \ll q$, which
is a contradiction. Thus
 the set $edge S_p \cap edge S_q$ has at least one point. Let us
 assume that $edge S_p \cap edge S_q$ has two different points $x$
 and $y$. Then, $x \rightarrow p$ and $p \rightarrow q$ implies that
 $x \leq q$. On the other hand, $x \in edge S_q$ implies that $x
 \rightarrow q$. In other words, two null geodesics $x \rightarrow
 p$ and $p \rightarrow q$ constitute a single null geodesic.

 Since $y \in edge S_p$, we have $y \rightarrow p$. If we follow the
 null geodesic from $y$ to $p$ and then follow null geodesic from
 $p$ to $q$, we have $y \ll q$ by the uniqueness of geodesic. This
contradicts to the fact that $y \in edge S_q$. This contradiction
stems from the assumption that $edge S_p \cap edge S_q$ has more
than one point. Thus, $edge S_p \cap edge S_q$ must have exactly one
element.
\end{proof}

We now prove the converse of the above proposition.

\begin{proposition} \label{horis2}
Let $p$ and $q$ be in $J^+(\Sigma)$ such that $p \neq q$. If $S_p
\subset S_q$ and $edge S_p \cap edge S_q $ has only one element,
then $p \rightarrow q$.
\end{proposition}
\begin{proof}
By Theorem \ref{causal}, $S_p \subset S_q$ implies $p \leq q$. Let
$x$ be the unique element in $edge S_p \cap edge S_q$. Then, we have
$x \rightarrow p$ and $x \rightarrow q$. If we assume that $p
\nrightarrow q$, then $p \ll q$ and thus we have $x \ll q$. This is
a contradiction to $x \rightarrow q$. Therefore we must have $p
\rightarrow q$.
\end{proof}

By combining the above two Propositions, we have the following
theorem.

\begin{thm} \label{horisthm}
Let $p$ and $q$ be two different points in $J^+(\Sigma)$.\\ Then $p
\rightarrow q$ if and only if $S_p \subset S_q$ and $edge S_p \cap
edge S_q$ has only one point.
\end{thm}

\section{Functions preserving causally admissible system } \label{section: 5}

In the previous sections, we have defined causally admissible slices
in Cauchy surface $\Sigma$ and investigated its properties in
connection with causal properties of $M$. In this section, we show
how the future admissible slices determine the structure of
$I^+(\Sigma)$. Since $I^+(\Sigma)$ is open in $M$, $I^+(\Sigma)$ is
a space-time in its own right.

Let $\mathcal{C}^+$ be the set of all future admissible slices of
$M$ with respect to $\Sigma$. That is to say, $\mathcal{C}^+ = \{
S_p = J^-(p) \cap \Sigma \,\, | \,\, p \in J^+(\Sigma) \}$ and we
call it future admissible system.

In the following, we review some known results and we apply the
results in previous sections to get new results.

\begin{definition}
A bijective function $f : M \rightarrow M^{\prime} $ is called a
causal isomorphism if $p \leq q$ $\Leftrightarrow$ $f(p) \leq f(q)$
and a chronological isomorphism if $p \ll q$ $\Leftrightarrow$ $f(p)
\ll f(q)$. If there exists a causal isomorphism (chronological
isomorphism, resp.) between $M$ and $M^\prime$, then we say that $M$
and $M^\prime$ are causally isomorphic (chronologically isomorphic,
resp.).
\end{definition}

The followings are well-known facts and can be found in Ref.
\cite{Malament}, \cite{Fullwood}, \cite{HKM}.

\begin{thm} \label{known}
Let $f : M \rightarrow M^\prime $ be a bijection between two
chronological space-times. Then the followings hold. \\
(i) $f$ is a causal isomorphism if and only if $f$ is a
chronological isomorphism. \\
(ii) If $f$ is a causal isomorphism, then $f$ is a homeomorphism. \\
(iii) If $f$ is a causal isomorphism, then $f$ is a smooth
conformal diffeomorphism.\\
\end{thm}

Let $M$ and $M^\prime$ be globally hyperbolic space-times with
non-compact Cauchy surfaces $\Sigma$ and $\Sigma^\prime$,
respectively. We assume that $\Sigma$ and $\Sigma^\prime$ be given
the corresponding future admissible systems $\mathcal{C}^+$ and
$\mathcal{C}^{\prime +}$, respectively, and we denote these by
$(\Sigma, \mathcal{C}^+)$ and $(\Sigma^\prime, \mathcal{C}^{\prime
+})$.

In view of Theorem \ref{known}, for two space-times $I^+(\Sigma)$
and $I^+(\Sigma^\prime)$ to be causally isomorphic, it is natural to
expect that there be a globally defined causal isomorphism between
$I^+(\Sigma)$ and $I^+(\Sigma^\prime)$. However, as the following
theorem shows, we only need a bijection between $\Sigma$ and
$\Sigma^\prime$ which preserves the structures of $\mathcal{C}^+$
and $\mathcal{C}^{\prime +}$.

\begin{thm} \label{subcentral1}
If $f : (\Sigma, \mathcal{C}^+) \rightarrow (\Sigma^\prime,
\mathcal{C}^{\prime +})$ is a bijection that satisfies the following
two conditions, then $I^+(\Sigma)$ and $I^+(\Sigma^\prime)$ are
causally isomorphic. \\
(i) For any $S \in \mathcal{C}^+$, $f(S) \in \mathcal{C}^{\prime +}$
and for any $S^\prime \in \mathcal{C}^{\prime +}$, there exists $S
\in \mathcal{C}^+$ such that $f(S) = S^\prime$.\\
(ii) For any $S_1$, $S_2$ in $\mathcal{C}^+$, we have $S_1 \subset
S_2$ if
and only if $f(S_1) \subset f(S_2)$.\\
\end{thm}
\begin{proof}
We show that any function $f$ that satisfies the two conditions can
be extended to a globally defined causal isomorphism not only
between $I^+(\Sigma)$ and $I^+(\Sigma^\prime)$ but between
$J^+(\Sigma)$ and $J^+(\Sigma^\prime)$.

Define $f^+ : J^+(\Sigma) \rightarrow J^+(\Sigma^\prime)$ as
follows. For given $p \in J^+(\Sigma)$, we have $S_p = J^-(p) \cap
\Sigma \,\,$ in $\,\, \mathcal{C}^+$. Then, by condition (i),
$f(S_p) \in \mathcal{C}^{\prime +}$. By definition of
$\mathcal{C}^{\prime +}$, there exists a unique point $p^\prime \in
J^+(\Sigma^\prime)$ such that $S_{p^\prime} = f(S_p)$ by Proposition
\ref{sp-unique}. Define $f^+ : J^+(\Sigma) \rightarrow
J^+(\Sigma^\prime)$ by letting $f^+(p) = p^\prime$. i.e. $S_{f^+(p)}
= f(S_p)$. Then by Proposition \ref{sp-unique}, $f^+$ is
well-defined and, since $f = f^+$ on $\Sigma$, $f^+$ is an extension
of $f$. We now show that $f^+$ is a causal isomorphism between
$J^+(\Sigma)$
and $J^+(\Sigma^\prime)$.\\

(i) $f^+$ is injective. :\\
If $f^+(p)=f^+(q)$, then $f(S_p) = f(S_q)$ by the property that
$S_{f^+(p)} = f(S_p)$. Since $f$ is a bijection, we have $S_p = S_q$
and thus $p = q$ by Proposition \ref{sp-unique}.\\

(ii) $f^+$ is surjective. : \\
For any $p^\prime \in J^+(\Sigma^\prime)$, we have $S_{p^\prime} \in
\mathcal{C}^{\prime +}$. By the condition (i), there exists $S_p \in
\mathcal{C}^+$ such that $f(S_p)=S_{p^\prime}$. Then by
definition of $f^+$, we have $f^+(p)=p^\prime$.\\

(iii) $f^+$ is a causal isomorphism. : \\
Let $p \leq q$ in $J^+(\Sigma)$. Then by Theorem \ref{causal}, we
have $S_p \subset S_q$. By the condition (ii), we have $f(S_p)
\subset f(S_q)$ and thus $S_{f^+(p)} \subset S_{f^+(q)}$. Then, by
Theorem \ref{causal}, $f^+(p) \leq f^+(q)$. Conversely, let us
assume that $f^+(p) \leq f^+(q)$ in $J^+(\Sigma^\prime)$. Then by
definition of $f^+$ and Theorem \ref{causal}, we have $S_{f^+(p)}
\subset S_{f^+(q)}$. By the condition (ii), $f(S_p) \subset f(S_q)$
implies that $S_p \subset S_q$. Then we have $p \leq q$ by
Theorem \ref{causal}.\\

(iv) $f^+(I^+(\Sigma)) = I^+(\Sigma^\prime)$. : \\
Since $\Sigma$ is a spacelike Cauchy surface, we have
$J^+(\Sigma)-\Sigma = I^+(\Sigma)$. Furthermore, since $f^+ :
J^+(\Sigma) \rightarrow J^+(\Sigma^\prime)$ is an extension of $f :
\Sigma \rightarrow \Sigma^\prime$ and both are bijections, we have
$f^+(I^+(\Sigma))=I^+(\Sigma^\prime)$.\\

This completes the proof.
\end{proof}

Since the above function $f$ which satisfies the given two
conditions gives us the simple criterion to causal isomorphism, we
give it the following definition.

\begin{definition}
If a bijection $f : \Sigma \rightarrow \Sigma^\prime$ between two
 non-compact Cauchy surfaces satisfies the above two conditions, then
we call $f$ a future admissible function and we denote it by $(M,
\Sigma, \mathcal{C}^+) \stackrel {f}{\cong} (M^\prime,
\Sigma^\prime, \mathcal{C}^{\prime +})$
\end{definition}

In Section \ref{section: 4}, we have seen the equivalent conditions
for two points to be chronologically, causally, and horismos related
in terms of their future admissible slices. At first glance, it
seems to be that the relations are independent. However, Theorem
\ref{known} and Theorem \ref{subcentral1} tell us the following.

\begin{corollary}
If $f : (\Sigma, \mathcal{C}^+) \rightarrow (\Sigma^\prime,
\mathcal{C}^{\prime +})$ is a future admissible function, then the
followings hold.\\
(i) If $S_1$ and $S_2$ in $\mathcal{C}^+$ are such that $S_1 \subset
S_2$ and $edge S_1 \cap edge S_2 = \emptyset$, then $f(S_1) \cap
f(S_2) = \emptyset$.\\
(ii) If $S_1$ and $S_2$ in $\mathcal{C}^+$ are such that $S_1
\subset S_2$ and $edge S_1 \cap edge S_2$ has one element, then
$edge f(S_1) \cap edge f(S_2)$ has one element.
\end{corollary}
\begin{proof}
Since $f$ is a future admissible function, the induced function $f^+
: I^+(\Sigma) \rightarrow I^+(\Sigma^\prime)$ is a causal
isomorphism. By Theorem \ref{known}, $f^+$ is also a chronological
isomorphism. Thus, if two future admissible slices $S_1$ and $S_2$
are such that $S_1 \subset S_2$ and $edge S_1 \cap edge S_2 =
\emptyset$, then their representative points are chronologically
related. Since $f^+$ is a chronological isomorphism, the image of
the representative points under $f^+$ must be also chronologically
related. Therefore the result follows from Theorem \ref{timelike}.
The second part can be proved in a similar way by Theorem
\ref{horisthm}.
\end{proof}

Although we have presented the materials in terms of future
admissible slices and future admissible functions, the same results
hold for past admissible slices and past admissible functions. For
the sake of completeness, we state the terminology and its
properties in the followings.

To distinguish future admissible slices and past admissible slices,
we denote future admissible slice by $S_p^+$ and past admissible
slice by $S_p^-$. To investigate the causal structure of
$J^-(\Sigma)$, we define $\mathcal{C}^-$ to be the set of all past
admissible slices of $M$ with respect to $\Sigma$. \\
i.e.$\mathcal{C}^- = \{ S_p^- = J^+(p) \cap \Sigma \,\, | \,\, p \in
J^-(\Sigma) \}$.

\begin{thm} \label{dual}
Let $p$ and $q$ be in $J^-(\Sigma)$. Then the followings hold.\\
(i) $p \leq q$ if and only if $S_q^- \subset S_p^-$. \\
(ii) $p \ll q$ if and only if $S_q^- \subset S_p^-$ and $edge S_p^-
\cap edge S_q^- = \emptyset$. \\
(iii) $p \rightarrow q$ if and only if $S_q^- \subset S_p^-$ and
$edge S_p^- \cap edge S_q^- $ has only one element.\\
\end{thm}

\begin{definition}
If a bijection $f : (\Sigma, \mathcal{C}^-) \rightarrow
(\Sigma^\prime, \mathcal{C}^{\prime -})$ between two non-compact
Cauchy surfaces satisfies the following two conditions,
then we say that $f$ is a past admissible function.\\
(i) For $S \in \mathcal{C}^-$, $f(S) \in \mathcal{C}^{\prime -}$ and
for any $S^\prime \in \mathcal{C}^{\prime -}$, there exists $S \in
\mathcal{C}^-$ such that $f(S) = S^\prime$.\\
(ii) For any $S_1$, $S_2$ in $\mathcal{C}^-$, we have $S_1 \subset
S_2$ if and only if $f(S_1) \subset f(S_2)$.\\
\end{definition}

We now state the dual form of Theorem \ref{subcentral1}, of which
the proof can be completed by the similar manner.

\begin{thm} \label{subcentral2}
If $f : (\Sigma, \mathcal{C}^-) \rightarrow (\Sigma^\prime,
\mathcal{C}^{\prime -})$ is a past admissible function, then
$I^-(\Sigma)$ and $I^-(\Sigma^\prime)$ are causally isomorphic.
\end{thm}

The proof of Theorem \ref{subcentral1} and Theorem \ref{subcentral2}
tell us that for $J^+(\Sigma)$ and $J^-(\Sigma)$ to be causally
isomorphic to $J^+(\Sigma^\prime)$ and $J^-(\Sigma^\prime)$,
respectively, we only need to compare the structures between
$\mathcal{C} = \mathcal{C}^+ \cup \mathcal{C}^-$ and
$\mathcal{C}^\prime = \mathcal{C}^{\prime +} \cup
\mathcal{C}^{\prime -}$ in which we identify $S_p^+$ and $S_p^-$ for
$p \in \Sigma$.

\begin{definition} For a given non-compact Cauchy surface $\Sigma$,
the set $\mathcal{C}=\mathcal{C}^+ \cup \mathcal{C}^-$ is called
causally admissible system on $\Sigma$.
\end{definition}

\begin{definition}
A bijection $f : \Sigma \rightarrow \Sigma^\prime$ between two
non-compact Cauchy surfaces is called a causally admissible function
if $f$ is both a future admissible function and a past admissible
function.
\end{definition}

We remark that a causally admissible function $f : \Sigma
\rightarrow \Sigma^\prime$ need not assumed to be continuous since
its definition and Theorem \ref{known} implies that $f$ must be a
homeomorphism as the following theorem shows.

\begin{thm} \label{central1}
Two space-times $M$ and $M^\prime$ with non-compact Cauchy surface
$\Sigma$ and $\Sigma^\prime$ are causally isomorphic if and only if
there exists a causally admissible function $f : (\Sigma,
\mathcal{C}) \rightarrow (\Sigma^\prime, \mathcal{C}^\prime)$
between the corresponding causally admissible systems.
\end{thm}
\begin{proof}
By Theorem \ref{subcentral1} and Theorem \ref{subcentral2}, a
causally admissible function $f : (\Sigma, \mathcal{C}) \rightarrow
(\Sigma^\prime, \mathcal{C}^\prime)$ can be extended to future
admissible function $f^+ : J^+(\Sigma) \rightarrow
J^+(\Sigma^\prime)$ and past admissible function $f^- : J^-(\Sigma)
\rightarrow J^-(\Sigma^\prime)$, which are causal isomorphisms.
Since $f^+=f^-=f$ on $\Sigma$, we can define $F : M \rightarrow
M^\prime$ by $F(p)=f^+(p)$ if $p \in I^+(\Sigma)$, $F(p)=f^-(p)$ if
$p \in I^-(\Sigma)$ and $F(p)=f(p)$ if $p \in \Sigma$. Since $F$ is
a causal isomorphism when restricted to $J^+(\Sigma)$ and
$J^-(\Sigma)$, it remains to show that for $p \in I^-(\Sigma)$ and
$q \in I^+(\Sigma)$, $p \leq q$ if and only $F(p) \leq F(q)$.

Assume that $p \leq q$ and let $\gamma$ be a causal curve from $p$
to $q$ in $M$. Since $M$ is globally hyperbolic, $\gamma$ must meet
$\Sigma$ at $x$, say. Since $x \leq q$, by Theorem \ref{causal}, we
have $S_x^+=\{x\} \subset S_q^+$ and likewise, by Theorem
\ref{dual}, we have $S_x^-=\{x\} \subset S_p^-$. Since $f$ is a
causally admissible function, we have $f(S_x^+) \subset f(S_q^+)$
and $f(S_x^-) \subset f(S_p^-)$. This implies that $f(S_q^+) \cap
f(S_p^-)$ is non-empty since it has $f(x)$ as its common element.
Therefore, we can conclude that $F(p) \leq F(q)$ since $F(p) \leq
F(x)$ and $F(x) \leq F(q)$. By following the same manner, we can
show that if $F(p) \leq F(q)$, then $p \leq q$.

We now assume that $M$ and $M^\prime$ be causally isomorphic and $g
: M \rightarrow M^\prime$ be their causal isomorphism. If we let
$\Sigma$ be a spacelike Cauchy surface of $M$, then $\Sigma^\prime =
g(\Sigma)$ is a smooth spacelike hypersurface of $M^\prime$ since
$g$ is a conformal diffeomorphism by Theorem \ref{known}. Let
$\gamma^\prime$ be an inextendible timelike curve in $M^\prime$,
then $g^{-1} \circ \gamma$ is also an inextendible timelike curve in
$M$. Since $\Sigma$ is a Cauchy surface in $M$, $g^{-1} \circ
\gamma^\prime$ must meet $\Sigma$ exactly once at $x$, say. Then
$\gamma^\prime$ meets $\Sigma^\prime$ exactly once at $g(x)$. Thus,
$\Sigma^\prime$ is a Cauchy surface of $M^\prime$. We now define a
causally admissible functions $f : \Sigma \rightarrow \Sigma^\prime$
by use of the causal isomorphism $g$. i.e. We let $f(x) = g(x)$ for
$x \in \Sigma$, then obviously $f$ is a bijection. It remains to
show that $f$ is, in fact, a causally admissible function. For $p
\in J^+(\Sigma)$, let $S_p^+ = J^-(p) \cap \Sigma$. Then, since $g$
is a causal isomorphism, we have $g(J^-(p)) = J^-(g(p))$. Thus,
$f(S_p^+) = g(J^-(p) \cap \Sigma) = g(J^-(p)) \cap g(\Sigma) =
J^-(g(p)) \cap \Sigma^\prime = S_{g(p)}^{+} \in \mathcal{C}^{\prime
+}$. Likewise, we can show that for any $S^\prime \in
\mathcal{C}^{\prime +}$, there exists $S \in \mathcal{C}^+$ such
that $f(S)=S^\prime$. If $S_1$ and $S_2$ in $\mathcal{C}^+$ are such
that $S_1 \subset S_2$, then $g(S_1) \subset g(S_2)$ implies that
$f(S_1) \subset f(S_2)$ and vice versa. By definition, $f$ is a
future admissible function. Likewise, we can show that $f$ is a past
admissible function and the proof is completed.
\end{proof}

The following Theorem, which can be obtained by combining Theorem
\ref{known} and Theorem \ref{central1}, states that when a
 space-time $M$ is given, its metric structure can be encoded into
its non-compact Cauchy surface by the family of compact subsets of
the Cauchy surface.

\begin{thm}
Let $M$ and $M^\prime$ be globally hyperbolic with non-compact
Cauchy surfaces $\Sigma$ and $\Sigma^\prime$. If there exists a
causally admissible function $f : (\Sigma, \mathcal{C}) \rightarrow
(\Sigma^\prime, \mathcal{C}^\prime)$ between the causally admissible
systems, then $M$ and $M^\prime$ are isometric up to a conformal
factor.
\end{thm}

From Theorem \ref{known}, we can see that causal structure
determines the metric structure of given space-time up to a
conformal factor. By the above Theorems, we can see that, for given
space-time $M$, its causally admissible system determines its causal
structure and the space-time metric up to a conformal factor. Since
causally admissible system consists of compact subsets of
non-compact Cauchy surface, we can say that all the causal
structure, chronological structure, differentiable structure and
metric structure up to a conformal factor can be encoded into its
Cauchy surface by causally admissible system.

If we see the non-compact Cauchy surface $\Sigma$ as the whole
``space" of our universe at one instant of time, then the above
theorems tell us that the ``space" and the collection of its compact
subspaces determines the whole structure of our universe. For
example, the structures of the  Robertson-Walker space-time model
with $\kappa = 0$ and $\kappa > 0$ can be determined from its
``spacelike" structures.

\section{acknowledgement}

This work was supported by the Korea Research Council of Fundamental
Science \& Technology (KRCF), Grant No. C-RESEARCH-2006-11-NIMS.

\end{document}